\documentclass{article}
\usepackage{amsfonts}
\usepackage{spconf,amsmath,graphicx,epsfig}
\usepackage{graphicx}
\usepackage{color}
\usepackage{microtype}
\usepackage{xcolor}
\usepackage{caption}
\usepackage{subcaption}
\usepackage{lineno}
\usepackage{tabularx}
\usepackage{units}
\usepackage{algorithm}
\usepackage{algpseudocode}
\usepackage{mathtools}
\usepackage{braket}
\usepackage{qcircuit}
\DeclarePairedDelimiter{\abs}{\lvert}{\rvert}

\usepackage{amsthm}
\newtheorem{theorem}{Theorem}

\newtheorem{lemma}[theorem]{Lemma}
\newtheorem{proposition}[theorem]{Proposition}

\title{Certified Robustness of Quantum Classifiers against Adversarial Examples through Quantum Noise}

\name{
\begin{tabular}{@{}c@{}}Jhih-Cing Huang$^{1}$\qquad Yu-Lin Tsai$^{2}$\qquad Chao-Han Huck Yang$^{3}$\qquad Cheng-Fang Su$^{2}$ \\Chia-Mu Yu$^{2}$\qquad Pin-Yu Chen$^{4}$\qquad Sy-Yen Kuo$^{1}$\end{tabular}}
\address{$^1$National Taiwan University, Taiwan\qquad $^2$National Yang Ming Chiao Tung University, Taiwan \\ $^3$ Georgia Institute of Technology, GA, USA\qquad $^4$IBM Research, NY, USA }

\usepackage{hyperref}

\begin{document}

\maketitle
\begin{abstract}
Recently, quantum classifiers have been found to be vulnerable to adversarial attacks, in which quantum classifiers are deceived by imperceptible noises, leading to misclassification. In this paper, we propose the \textbf{first theoretical study} demonstrating that \textbf{adding quantum random rotation noise can improve robustness} in quantum classifiers against adversarial attacks. We link the definition of differential privacy and show that the quantum classifier trained with the natural presence of additive noise is differentially private. Finally, we derive a certified robustness bound to enable quantum classifiers to defend against adversarial examples, supported by experimental results simulated with noises from IBM's 7-qubits device.
\end{abstract}
\section{Introduction}
\label{sec:introduction}

The joint study of quantum computing and machine learning opens a new research area named quantum machine learning \cite{preskill2018quantum}.
For example, quantum classifiers \cite{li2022recent, qi2022classical, yang2021decentralizing, biamonte2017quantum, yang2022bert} solve classification problems similar to the classical ones. Furthermore, the data availability of quantum classifiers is wider since quantum data, which is generated from natural or artificial quantum systems, can also be included.
In particular, quantum neural network (QNN) is a pure quantum model~\cite{benedetti2019generative, dallaire2018quantum, farhi2018classification, huggins2019towards, schuld2020circuit} that contains trainable parameters. Thus, QNN is also identified as a parameterized quantum circuit. Aside from classical machine learning, QNNs predict labels of data by measuring the ancilla qubits that are carried through along with the training data. Based on the training data, QNNs utilize optimizers similar to the classical ones to optimize parameters with specific techniques \cite{mitarai2018quantum, beer2020training, qi2022theoretical} to calculate the gradient. However, recent findings also suggest that the QNNs-based model is also sensitive to small gradient-based perturbation with malicious misclassification results. To deal with the adversarial examples on QNNs, Weber et al. \cite{optimal} find out a tight condition for the robustness against the adversarial perturbation for QNNs. Furthermore, Du et al. \cite{PhysRevResearch.3.023153} propose that QNN can resist against the adversarial perturbation via depolarization noise. 

\textbf{Contribution} In this paper, inspired by \cite{PhysRevResearch.3.023153}, we propose one  \textbf{formal theoretical analysis on robustness of QNNs}, where the robustness can be improved via the added quantum random rotation noise. The quantum rotation noise in our theoretical solution additionally applied to QNNs enables us to derive a certified robustness bound.

\subsection{Related Work}

\indent \textbf{Randomized Smoothing}
Classical randomized smoothing \cite{cohen2019certified} protects classical classifiers through adding noise, such as Gaussian Noise, and by sampling multiple times, the label with the highest probability is assigned to the data point; i.e., $g(x)=\underset{c\in \kappa}{\arg\max} P(f(x+\epsilon)=c)$, where $\kappa$ is the set of all possible classes and $\epsilon \sim N(0,\sigma^2 I)$.

\textbf{Differential Privacy}
Differential privacy \cite{dwork2006calibrating} is one of the \textit{de facto} standard in the realm of data privacy. Furthermore, differential privacy quantifies the amount of privacy protection. Formally, we can define the notion $\epsilon$-differential privacy in the following. Let $\epsilon$ be a positive real number and ${\displaystyle {\mathcal {A}}}$ be a randomized algorithm that takes a dataset as input. Let ${\displaystyle {\textrm {im}}\ {\mathcal {A}}}$ denote the image of ${\displaystyle {\mathcal {A}}}$. The algorithm ${\displaystyle {\mathcal {A}}}$ is said to provide ${\displaystyle \epsilon }$-differential privacy if, for all datasets ${\displaystyle D_{1}}$ and $ {\displaystyle D_{2}}$ that differ on a single element (i.e., the data of one person), and all subsets ${\displaystyle S}$ of $ {\displaystyle {\textrm {im}}\ {\mathcal {A}}}$, we have $\Pr[{\mathcal {A}}(D_{1})\in S]\leq e^{\epsilon}\cdot \Pr[{\mathcal {A}}(D_{2})\in S]$,
where the probability is taken over the randomness used by the algorithm.

\textbf{Certified Robustness} For a certifiably robust~\cite{zhang2018efficient, yang2021causal} classical classifier with robustness bound $d$, the predictions of input data points $x$ and $x^{'}$ are guaranteed to be the same where $x$ and $x^{'}$ are neighboring data points with $l_i$ distance of $x$ and $x^{'}$ less than threshold $d$. The definition of certified robustness of quantum classifiers with robustness bound $\tau_D$ is similar to classical one. 

Namely, for two quantum states $\sigma$ and $\rho$, if $||\sigma-\rho|| < \tau_D$, where the distance is defined by the trace distance(i.e. $||\sigma-\rho|| := \frac{Tr(||\sigma-\rho||)}{2}$), the majority of output labels of $\sigma$ and $\rho$ are guaranteed to be the same.

\textbf{QNN Robustness} Recently, some research efforts have been on developing the robustness for QNNs. For example, Weber et al. \cite{PhysRevResearch.4.033217} formulate this adversarial robustness topic on QNNs as semidefinite programs and considers higher statistical moments of the observable and generalized bounds. Weber et al. \cite{optimal} establsih a link between binary quantum hypothesis testing and provably robust QNNs, resulting in a robustness condition for the amount of noise a classifier can tolerate. Du et al. \cite{PhysRevResearch.3.023153} finds that the depolarization noise in QNNs helps derive a robustness bound, where the robustness improves with increasing noise.

\section{Background Knowledge}
\textbf{Quantum Classifier} Parameterized quantum circuits are quantum frameworks that depend on trainable parameters, and can therefore be optimized \cite{benedetti2019parameterized}. Variational quantum classifier algorithm is under this framework, being the predominant basis of quantum classifiers. Several optimization methods are developed and different quantum classifiers are proposed. In our work, we use the optimization method called parameter-shift rule \cite{mitarai2018quantum}. That is, the output of a variational quantum circuit, denoted by $f(\theta)$, is parameterized by $\theta=\theta_1,\theta_2,\dots$.

To optimize $\theta$, we need to acquire the partial derivative of $f(\theta)$ which can be expressed as a linear combination of other quantum functions, typically derived from the same circuit with a shift of $\theta$. That is, the same variational quantum circuit can be used to compute the partial derivatives of  $f(\theta)$. Besides, we need to encode our classical data and int this aspect we adopt the amplitude encoding method \cite{schuld2017implementing, mottonen2004transformation} . To encode data efficiently, amplitude encoding is to transform classical data into linear combination of independent quantum states with the magnitudes of features being the weights which can be expressed as $S_x\ket{0}=\frac{1}{\abs{x}}\sum_{n=1}^{2^n}x_i\ket{i}$, where each $x_i$ is a feature (component) of data point $x$, and ${\ket{i}}$ is a basis of $n$-qubit space. In our work, we assume a $K$-class quantum classifier of which the output is the predicted label of the input state.

Let $\Pi_k$ be a positive operator-valued measure (POVM) and $\mathcal{E}$ be quantum operations of the quantum classifier. Define $y_k(\sigma) \equiv Tr(\Pi_k\mathcal{E} (\sigma \otimes \ket{a}\bra{a}))$ which denotes the probability with which input state $\sigma$ is assigned to the label $k, k \in {0,1,2,....,K-1}.$ and  $\overset{\sim}{y_k}(\sigma) = Tr(\Pi_k\mathcal{E} (R\sigma \otimes \ket{a}\bra{a}))$ which denotes the probability with which input state $\sigma$ is assigned to the label $k, k \in {0,1,2,....,K-1}$ under noise, where $R$ is the noise operator. Since it is impossible to derive actual $y_k(\sigma)$ and $\overset{\sim}{y_k}(\sigma)$, we sample N times to estimate $y_k$ and $\overset{\sim}{y_k}$ with $y_k^{(N)}(\sigma)$ and $\overset{\sim}{y_k}^{(N)}(\sigma)$, respectively. In our work, we assume $K=2$ (binary classification) for convenience but similar reasoning can also be applied in the multiclass scenario.

\textbf{Quantum Differential Privacy}
Similar to classical $\epsilon$-differential privacy, we adopt the quantum version of $\epsilon$-differential privacy from \cite{zhou2017differential}. Furthermore, we express a quantum classifier under $K$-class classification problem to have satisfied $\epsilon$-differential privacy if the following holds.

Let $\epsilon$ be a positive real number and ${\displaystyle {\mathcal {M}}}$ be a quantum algorithm that takes a quantum state as input. The algorithm ${\displaystyle {\mathcal {M}}}$ is said to provide quantum ${\displaystyle \epsilon}$-differential privacy if, for all input quantum states $\sigma$ and $ \rho$ such that $\tau(\sigma,\rho)<\tau_D$, and for all $\Pi_i, i \in {0,1,2,3...,K-1}$, we have ${\displaystyle \Pr[{\mathcal {M}}(\sigma,\Pi_i)]\leq \exp \left(\epsilon \right)\cdot \Pr[{\mathcal {M}}(\rho,\Pi_i)]}$, and therefore $e^{-\epsilon} \leq \frac{\overset{\sim}{y_k}(\rho)}{\overset{\sim}{y_k}(\sigma)} \leq e^{\epsilon}$.

\section{Proposed Method}
We begin with the idea of simulating randomized smoothing in quantum machine learning. We aim to add perturbation on qubits and consider random rotations on the Bloch sphere as a counterpart of randomized smoothing. Then, we apply rotation gates, as shown in Fig. \ref{fig:circuit}, on each input qubit and set up rotation angles with random variables generated from classical computers.

\begin{figure}[ht]
\centering
        \includegraphics[scale=0.2]{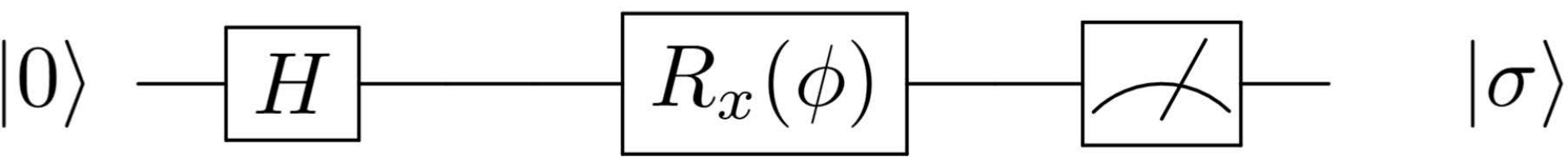}
        \caption{The rotation circuit with output density matrix ($\sigma$). }
\label{fig:circuit}
\end{figure}
\vspace{-2mm}

Our proposed method is summarized in Algorithm~\ref{tab:alg}. Our method does not assume details of quantum classifiers, and thus is general for model agnostic. Our method guarantees the accuracy of original quantum classifiers, and the corresponding robustness bound is also applicable for all kinds of quantum classifiers. Further analysis of Algorithm~\ref{tab:alg} is also proven in subsequent section.

\begin{algorithm}
    \hspace*{\algorithmicindent} \textbf{Input} $\sigma$: where $\sigma$ is density matrix of n-dim data.  \\
    \hspace*{\algorithmicindent} \textbf{Output}
    $f(\theta^*, \sigma)$
    \begin{algorithmic}[H]
    \State 1. For a chosen quantum classifier, add Pauli-X operators before each input qubit.
    \State 2. Generate n random variables $\theta_{1}, \theta_{2}, ..., \theta_{n}$ subject to $0<h_1 < \tan\theta_i < h_2$ for all $i\in\{1, 2, \dots, n\}$.
    \State 3. Set up rotation angles of additional Pauli-X operators with $\theta_{1}, \theta_{2}, ..., \theta_{n}$
    \State 4. Execute the quantum classifier $N$ times to get the score vector $f(\theta^*, \sigma)$.
    \end{algorithmic}
    \vspace*{1em}
    \caption{Quantum model under quantum noise rotation}
    \label{tab:alg}
\end{algorithm}
\vspace{-2mm}
\section{Theoretical Analysis}
\label{Analysis}
Our goal is to demonstrate that random rotation noises can be used to protect quantum classifiers against adversarial perturbations. This can be divided into three main steps. We first show the invariance of outcomes between noisy classifiers and original ones. Then we demonstrate how random rotation noises improve quantum differential privacy for the classifiers. Ultimately, we can show the connection between the differential privacy and the better robustness against general adversaries of classifiers.

\subsection{Accuracy of Noisy Classifiers}
In this section, we will start with the derivation of the accuracy given the noisy quantum classifiers whose output will not be affected under a known noise source; therefore, we might use it for the improvement of the privacy. We should emphasize that robustness bounds against an unknown adversary belong to the worst-case scenario, while those against a commonly known noise source are far from the worst-case scenario. 
\begin{lemma}
\label{lemma1}
Let $\sigma$ denote a input quantum state and $y_k(\sigma)$ denote the output for the noiseless circuit. Then for binary case, if the class label $C$ is assigned to $\sigma$ by the noiseless circuit with $y_C(\sigma) > \frac{(1+h)^n}{2}$, i.e., $C = \underset{k}{argmax}\hspace{1mm} y_k(\sigma)$, then the same label is also assigned by the noisy circuit. This means $ \underset{k}{argmax} \hspace{1mm} \overset{\sim}{y_k}(\sigma) = C $ for any $\sigma$. %
\end{lemma}

The above result gives the robustness of quantum classifiers against random rotation noises if the theoretical probabilities $\overset{\sim}{y_k}(\sigma)$ has been accessed. Nevertheless, we can only sample the circuit finite times and obtain the estimated values $\overset{\sim}{y_k}^{(N)}(\sigma)$ for a N-time sampling. In order to guarantee the robustness against our noise to high probability, we then give the relation between required sampling complexity N and the magnitude of added noises related to its superposition. 

\emph{Proof.} Due to the limited space, we provide a proof of Lemma 1 in our Appendix~\ref{apx1}.

\begin{proposition}
\label{prop1}
Let the predicted classification label of $\sigma$ using the noiseless $K$-multiclass classification circuit be $C$.
This means we can define $\xi \equiv y_C(\sigma) - \underset{k\ne C}{max} \hspace{1mm} y_k(\sigma)$
where $\xi > 0$. In the corresponding noisy circuit, one samples the circuit $N$ times for each $k$ to obtain the estimates $\overset{\sim}{y_k}^{(N)}(\sigma)$. Then
$\sigma$ is also labelled C with probability at least $\beta$ if the sample complexity $ N \sim \frac{1}{8(\epsilon-(1+h)^n+1)^2}\ln(\frac{2}{1-\beta})$\\
\end{proposition}

\begin{proof}
We begin with Hoeffding inequality,
\begin{align*}
    Pr(|\frac{1}{N}\sum_{i=1}^{N} Z_i - E(Z_i)| \leq \zeta) \geq 1-2\exp(-\frac{2N\zeta^2}{(b-a)^2}),
\end{align*}
where $Z_i \in [a,b]$.\\
Here, $a=0$, $b=1$, $\zeta = 2\eta$, $\frac{1}{N}\sum_{i=1}^{N} Z_i = \overset{\sim}{y_k}^{(N)}(\sigma)$ and $E(Z_i)=\overset{\sim}{y_k}$. If we require the probability $Pr(|\overset{\sim}{y_k}^{(N)}(\sigma)-\overset{\sim}{y_k}| < 2\eta) \geq \beta$, it is sufficient to require \\
$N \sim \frac{1}{8\eta^2}\ln(\frac{2}{1-\beta})$, where
$\eta = \overset{\sim}{y_1}(\sigma) - \overset{\sim}{y_2}(\sigma) > g(\xi-(1+h)^n+1)$ following the result of Lemma \ref{lemma1}. Therefore, it implies that
\begin{align}
N \sim \frac{1}{8(\xi-(1+h)^n+1)^2}\ln(\frac{2}{1-\beta}).
\end{align}
\label{proof_lemma}
\end{proof}

\subsection{Relation between Noise Magnitude and Quantum Differential Privacy}
\label{sec:analysis 2}

By adding noise to quantum classifiers, the trained model will obey differential privacy. We now prove the relation between noise magnitude and quantum $\epsilon$-differential privacy.
\begin{lemma}
\label{lemma2}
Let the algorithm $M$ correspond to the $K$-multiclass classification circuit with random rotation noise channels and measurement operators $\{\Pi_k\}_{k=1}^{K}$. Then for two quantum test states $\sigma$ and $\rho$ obeying $\tau(\rho, \sigma) \leq \tau_D$ with $0 \leq \tau_d \leq 1$, $M$ satisfies $\epsilon$-quantum differential privacy where
$\epsilon = \ln{(1+\frac{\tau_D}{t^n})}$.\\
\end{lemma}

\begin{proof}
If the relation,
$e^{-\epsilon} \leq \frac{\overset{\sim}{y_k}(\rho)}{\overset{\sim}{y_k}(\sigma)}\leq e^{\epsilon}$ holds for every $k$, it can be derived that
\begin{align*}
    \frac{\overset{\sim}{y_k}(\rho)}{\overset{\sim}{y_k}(\sigma)}-1&=\frac{\overset{\sim}{y_k}(\rho)-\overset{\sim}{y_k}(\sigma)}{\overset{\sim}{y_k}(\sigma)}\\
    &=\frac{\overset{\sim}{y_k}(\rho-\sigma)}{\overset{\sim}{y_k}(\sigma)}\leq \frac{\tau_D T_r(\pi_k)}{t^n}=e^{\epsilon}-1
\end{align*}
with the inequality $Tr(U(\rho-\sigma)U^{\dagger}\Lambda_k) \leq \tau_DTr(\Lambda_k)= \tau_D$. The result above implies 
$\epsilon=\ln{(1+\frac{\tau_D}{t^n }})$.
\end{proof}

From Lemma \ref{lemma1}, we can find that the \textbf{privacy budget $\epsilon$ of the quantum classifier is inversely proportional to the lower bound of noise magnitude}.

\subsection{Connection between Quantum Differential Privacy and Certified Robustness}
Following the result of Lemma \ref{lemma1} and Lemma \ref{lemma2}, we proceed to derive the relationship between quantum differential privacy and certified robustness.
\begin{theorem}
\label{theo1}
(Infinite sampling case). We begin with our $K$-multiclass classification circuit with random rotation noise.
Let infinite sampling of the output be allowed, we can find
$\overset{\sim}{y_k}(\rho)$ for $k = 0, ..., k-1$ for any test state $\rho$ given. Suppose $\overset{\sim}{y_C}(\sigma) > e^{2\epsilon} max_{k \ne C}\overset{\sim}{y_k}(\sigma) $ holds, where $\epsilon = \ln{(1+\frac{\tau_D}{t^n})}$, which implies that $\sigma$ is assigned the class label $C$, i.e.,  $C = \underset{k}{argmax} \hspace{1mm} \overset{\sim}{y_k}(\sigma) = \underset{k}{argmax}\hspace{1mm} y_k(\sigma)$
Then $\rho$ is also labelled as $C$, i.e., $C = \underset{k}{argmax} \hspace{1mm} \overset{\sim}{y_k}(\rho) = \underset{k}{argmax}\hspace{1mm} y_k(\rho)$ for any $\rho$ where $\tau(\rho, \sigma) \leq \tau_D$ with $0 \leq \tau_d \leq 1$.
\end{theorem}

That is, for two quantum states $\sigma$ and $\rho$ where $\rho$ is an adversarial example and $\tau(\sigma,\rho) \leq \tau_D$, the predicted label of $\rho$ will be identical to that of $\sigma$. When noise magnitude increases, the range of adversarial examples which can be defended will be enlarged, and therefore, robustness of the classifier will have been enhanced.
The following theorem shows a similar result still holds for finite sampling cases. 
\begin{theorem}
\label{theo2}
(Finite sampling case). Suppose one samples
the output of the circuit $N$ times for the estimation of each
$\underset{k}{argmax} \hspace{1mm} \overset{\sim}{y_k}(\sigma)$. Let $\overset{\sim}{y_C}(\sigma) - \zeta > e^{2\epsilon} max_{k \ne C}(\overset{\sim}{y_k}(\sigma)+\zeta) $  where
$\epsilon = \ln{(1+\frac{\tau_D}{t^n})}$ , which implies $\sigma$ has the
class label $C$. Then the class label of $\rho$ is also $C$, i.e.,$C = \underset{k}{argmax} \hspace{1mm} \overset{\sim}{y_k}(\rho) = \underset{k}{argmax}\hspace{1mm} y_k(\rho)$ for any $\rho$ to probability at
least $1 - 2 exp(-2N\zeta^2)$ for any $\rho$ where $0 \leq \tau_d \leq 1$. This
also implies $\overset{\sim}{y_C}(\sigma) + \zeta >  max_{k \ne C}(\overset{\sim}{y_k}(\sigma)-\zeta) $ to probability at least $1 - 2 exp(-2N\zeta^2)$. 
\end{theorem}

Accordingly, we can utilize above lemmas and theorems to find $\tau_D$ which represents the range of durable adversarial perturbation.

\begin{theorem}
\label{theo3}
Define $B \equiv \frac{\overset{\sim}{y_0}(\sigma)}{\overset{\sim}{y_1}(\sigma)}$. With random rotation noise, the classifier is robust under any adversarial perturbation $\sigma \rightarrow \rho$ such that $\tau(\sigma,\rho) \leq \tau_D$, if $B>e^{2\epsilon}$.
\end{theorem}

Ultimately, for a binary classifier, if $\frac{\overset{\sim}{y_0}(\sigma)}{\overset{\sim}{y_1}(\sigma)} > e^{2\epsilon}$, both Theorem \ref{theo1} and Theorem \ref{theo2} are applicable. Hence, we can derive certified robustness of the classifier by deriving $\tau_D$.
\begin{proof}
From Lemma \ref{lemma2}, we prove that the noisy classifier satisfies $\epsilon$-quantum differential privacy where $\epsilon=\ln(1+\frac{\tau_D}{t^n})$. It implies that
\begin{align*}
     e^{-\epsilon}\overset{\sim}{y_0}(\sigma)\leq\overset{\sim}{y_0}(\rho)\leq e^{\epsilon}\overset{\sim}{y_0}(\sigma)\\
e^{-\epsilon}\overset{\sim}{y_1}(\sigma)\leq\overset{\sim}{y_1}(\rho)\leq e^{\epsilon}\overset{\sim}{y_1}(\sigma).
\end{align*}

Without loss of generality, we assume for adversarial example $\rho$, $ C=0$, that is, $\underset{k}{argmax} \hspace{1mm} \overset{\sim}{y_k}(\rho)=0$. It then can be showed that
$\overset{\sim}{y_0}(\rho)>\overset{\sim}{y_1}(\rho)=1-\overset{\sim}{y_0}(\rho)$,
which implies that
\begin{align*}
     \frac{\overset{\sim}{y_0}(\sigma)}{\overset{\sim}{y_1}(\sigma)}>e^{2\epsilon}=(1+\frac{\tau_D}{t^n})^2.
\end{align*}

Here, we define 
$B \equiv \frac{\overset{\sim}{y_0}(\sigma)}{\overset{\sim}{y_1}(\sigma)}$ which can be obtained from experiments. With the definition of $B$, equation above turns into
\begin{align*}
    \tau_D < (\sqrt{B}-1)t^n.
\end{align*}
\end{proof}

\section{Experimental Results}
In our experiment, the dimension of input data is $3$ following the setting in~\cite{PhysRevResearch.3.023153}. We encode each input data into a three-qubit quantum state. For the noise magnitude setting, we set the lower bound of noise magnitude equal to zero, that is, $t=0$, and we generate three random variables from the uniform distribution $U(0,h)$ to randomize the quantum state. 
        
\begin{figure}[ht]
\centering
     \includegraphics[scale=0.34]{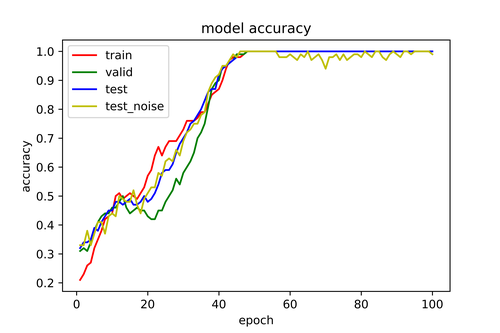}
        \caption{Model accuracy with varying noise magnitudes.}
\label{fig.magnitude1}
\end{figure}

From Fig.~\ref{fig.magnitude1}, the difference in testing accuracy (\textcolor{blue}{noiseless}/\textcolor{orange}{noisy}) with the same sample size is not obvious, where the performance gap will decrease if we enlarge the sample size. That is, the impact of noise on accuracy is negligible, and the model will be certified robust with the robustness bound deduced from Theorem \ref{theo3}. 

\begin{figure}[ht]
\centering
        \includegraphics[scale=0.34]{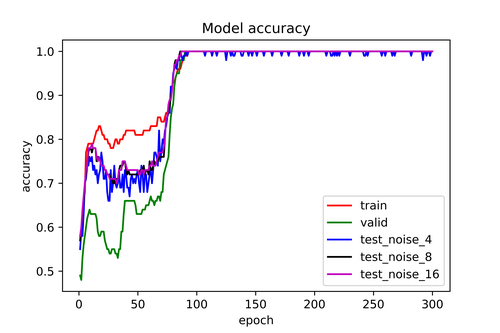}
        \caption{Model accuracy with varying noise magnitudes.}
\label{fig.magnitude2}
\end{figure}

Fig.~\ref{fig.magnitude2} shows the model accuracy with varying noise magnitudes. ``Test\_noise\_4,'' ``test\_noise\_8,'' and ``test\_noise\_16'' denote the testing accuracy with $h$ equal to $\frac{2\pi}{2^4}$, $\frac{2\pi}{2^8}$, and $\frac{2\pi}{2^{16}}$ respectively. The upper bound of noise magnitude will impact quantum classifiers' stability, but similarly, the impact can be diminished by enlarging the sample size.

\section{Conclusion}
In this research, we theoretically prove that by applying random rotation noise to input quantum states, the quantum classifier are able to defend against adversarial attack. In the theoretical development, we do not presume any architecture details of quantum classifiers. In other words, as long as the noise presence is accompanied in the training process, then any trained quantum classifier satisfies the above robustness guarantee. Besides, our method doesn't require to retrain classifiers, and the procedure simply needs to set rotation angles of quantum gates with random variables generated by a classical computer. This means that the computation complexity after adopting our algorithm is essentially the same as the original one. We will release our implementation for future study. 

\clearpage
\bibliographystyle{IEEEtran}
\bibliography{ref}
\clearpage
\appendix

\section*{Appendix}

\subsection*{A. Proof of Lemma~\ref{lemma1}}
\label{apx1}
\begin{proof}
Suppose that we add quantum rotation noise to the quantum classifier with quantum test state $\sigma$ and quantum gates $g= \prod\limits_{x_i}^{n} \cos\theta_{x_i}$. Given distinct $n$ numbers $\theta_{x_1}, \theta_{x_2}, \dots, \theta_{x_n}$ that satisify $0<h_1<\tan\theta_{x_i}<h_2$ for all $i\in\{1, 2, \dots, n\}$, then the corresponding probability of class K becomes the superposition of other input states, that is,\\
\begin{align}
 \label{eq:1}
    \overset{\sim}{y_k}(\sigma)&=g\hspace{0.5mm}y_k (\sigma)+g\sum_{{x_1}=1}^{n}\tan\theta_{x_1}y_k(\sigma_{x_1})\notag \\
    &\quad +g\sum_{x_1, x_2=1}^{n} \tan\theta_{x_1}\tan\theta_{x_2}y_k(\sigma_{x_1})y_k(\sigma_{x_2})+\cdots\notag \\
    &= gy_k(\sigma)+g\sum_{\ell=1}^n\{\prod_{i=1}^\ell\tan\theta_{x_i}y_k(\sigma_{x_i})\}.
\end{align}

Our objective is to prove that $\underset{k}{argmax}\hspace{1mm} \overset{\sim}y_k(\sigma) = C$.
Here we consider binary case. Without loss of generality, we assume that  $\hspace{1mm}\overset{\sim}y_1(\sigma)>\overset{\sim}y_2(\sigma)$. From equation \ref{eq:1}, we formulate the probabilities of label 1 and label 2 respectively, i.e.,

\begin{align}
 \label{eq:2}
    \overset{\sim}{y_1}(\sigma)=g\hspace{0.5mm}y_1(\sigma)+g\sum_{\ell=1}^n\{\prod_{i=1}^\ell\tan\theta_{x_i}y_1(\sigma_{x_i})\}.
\end{align}
\begin{align}
     \label{eq:3}
\overset{\sim}{y_2}(\sigma)=g\hspace{0.5mm}y_2(\sigma)+g\sum_{\ell=1}^n\{\prod_{i=1}^\ell\tan\theta_{x_i}y_2(\sigma_{x_i})\}.
\end{align}
Here, the relation, $\overset{\sim}{y_1}(\sigma)>\overset{\sim}{y_2}(\sigma)$, shall be required to prove our objective. With equation \ref{eq:2} and \ref{eq:3}, it can be rewritten to be 
\begin{align*}
     &g\hspace{0.5mm}y_1(\sigma)+g\sum_{\ell=1}^n\{\prod_{i=1}^\ell\tan\theta_{x_i}y_1(\sigma_{x_i})\}\\
     &\quad>g\hspace{0.5mm}y_2(\sigma)+g\sum_{\ell=1}^n\{\prod_{i=1}^\ell\tan\theta_{x_i}y_2(\sigma_{x_i})\}.
\end{align*}
We can further rearrange the inequality to be
\begin{align*}
\label{eq:4}
         y_1(\sigma)-y_2(\sigma)>&\sum_{\ell=1}^n\{\prod_{i=1}^\ell\tan\theta_{x_i}y_2(\sigma_{x_i})\} \\ &-\sum_{\ell=1}^n\{\prod_{i=1}^\ell\tan\theta_{x_i}y_1(\sigma_{x_i})\}.
\end{align*}
Consider worst case :  All\hspace{1mm}$ y_2(\sigma_{x_i})=1$, which leads that
\begin{align}
    \sum_{\ell=1}^n\{\prod_{i=1}^\ell\tan\theta_{x_i}y_1(\sigma_{x_i})\}=0
\end{align}
With the rearrangement,
\begin{align*}
    \sum_{\ell=1}^n\{\prod_{i=1}^\ell\tan\theta_{x_i}y_2(\sigma_{x_i})\}=\prod_{i=1}^{n}(1+\tan\theta_{x_i}y_2(\sigma_{x_i}))-1,
\end{align*}
we obtain the constrain, 
\begin{align*}
    (1+h_1)^n-1 &\leq
    \sum_{\ell=1}^n\{\prod_{i=1}^\ell\tan\theta_{x_i}y_2(\sigma_{x_i})\}\\
    &\leq(1+h_2)^n-1.
\end{align*}
Equation \ref{eq:4} becomes
\begin{align}
\label{eq:5}
    y_1(\sigma) - y_2(\sigma) \geq (1+h_1)^n-1.
\end{align}
Combining equation \ref{eq:5} and the property,
$1=y_1(\sigma)+y_2(\sigma)$, we can derive that
\begin{align*}
     2y_1(\sigma)>(1+h_1)^n \Rightarrow y_1(\sigma)>\frac{(1+h_1)^n}{2}.
\end{align*}
\\
Therefore, if there is a positive number $h$ satisfies  
$$
y_C(\sigma)>\frac{(1+h)^n}{2},
$$
then  $C=\underset{k}{argmax}\hspace{1mm}\overset{\sim}{y_k}(\sigma)$.
\end{proof}

\end{document}